\documentclass[11pt,letterpaper]{article}
\usepackage{lmodern}
\usepackage[T1]{fontenc}
\usepackage[utf8]{inputenc}
\usepackage{sty}

\title{Counterexamples to the Low-Degree Conjecture}

\usepackage{authblk}
\author[1]{Justin Holmgren\thanks{Email: \texttt{holmgren@alum.mit.edu}.}}
\author[2]{Alexander S.\ Wein\thanks{Email: \texttt{awein@cims.nyu.edu}. Partially supported by NSF grant DMS-1712730 and by the Simons Collaboration on Algorithms and Geometry.}}

\affil[1]{Simons Institute for the Theory of Computing}

\affil[2]{Department of Mathematics, Courant Institute of Mathematical Sciences, New York University}

\begin{document}

\maketitle

\begin{abstract}
A conjecture of Hopkins~\cite{sam-thesis} posits that for certain high-dimensional hypothesis testing problems, no polynomial-time algorithm can outperform so-called ``simple statistics'', which are low-degree polynomials in the data. This conjecture formalizes the beliefs surrounding a line of recent work that seeks to understand statistical-versus-computational tradeoffs via the \emph{low-degree likelihood ratio}. In this work, we refute the conjecture of Hopkins~\cite{sam-thesis}. However, our counterexample crucially exploits the specifics of the noise operator used in the conjecture, and we point out a simple way to modify the conjecture to rule out our counterexample. We also give an example illustrating that (even after the above modification), the symmetry assumption in the conjecture is necessary. These results do not undermine the low-degree framework for computational lower bounds, but rather aim to better understand what class of problems it is applicable to.
\end{abstract}

\section{Introduction}

A primary goal of computer science is to understand which problems can be solved by efficient algorithms. Given the formidable difficulty of proving unconditional computational hardness, state-of-the-art results typically rely on unproven conjectures. While many such results rely only upon the widely-believed conjecture $\complexityP \neq \NP$, other results have only been proven under stronger assumptions such as the unique games conjecture \cite{khot-ugc,ugc-survey}, the exponential time hypothesis~\cite{ImpagliazzoP01}, the learning with errors assumption~\cite{Regev09}, or the planted clique hypothesis \cite{jerrum-clique,BR-reduction}.

It has also been fruitful to conjecture that a specific \emph{algorithm} (or limited class of algorithms) is optimal for a suitable class of problems. This viewpoint has been particularly prominent in the study of average-case noisy statistical inference problems, where it appears that optimal performance over a large class of problems can be achieved by methods such as the sum-of-squares hierarchy (see \cite{sos-survey}), statistical query algorithms \cite{Kearns93,BlumFJKMR94}, the approximate message passing framework \cite{AMP,LKZ-AMP}, and low-degree polynomials \cite{HS-bayesian,sos-hidden,sam-thesis}. It is helpful to have such a conjectured-optimal meta-algorithm because this often admits a \emph{systematic} analysis of hardness. However, the exact class of problems for which we believe these methods are optimal has typically not been precisely formulated. In this work, we explore this issue for the class of low-degree polynomial algorithms, which admits a systematic analysis via the low-degree likelihood ratio.

The \emph{low-degree likelihood ratio} \cite{HS-bayesian,sos-hidden,sam-thesis} has recently emerged as a framework for studying computational hardness in high-dimensional statistical inference problems. It has been shown that for many ``natural statistical problems,'' all known polynomial-time algorithms only succeed in the parameter regime where certain ``simple'' (low-degree) statistics succeed. The power of low-degree statistics can often be understood via a relatively simple explicit calculation, yielding a tractable way to precisely predict
the statistical-versus-computational tradeoffs in a given problem. These ``predictions'' can rigorously imply lower bounds against a broad class of spectral methods~\cite[Theorem~4.4]{lowdeg-notes} and are intimately connected to the \emph{sum-of-squares hierarchy} (see~\cite{sos-hidden,sam-thesis,sos-survey}). Recent work has (either explicitly or implicitly) carried out this type of low-degree analysis for a variety of statistical tasks \cite{p-cal,HS-bayesian,sos-hidden,sam-thesis,graph-matching,sk-cert,lowdeg-notes,subexp-sparse,BB,lifting-sos,heavy-tail}. For more on these methods, we refer the reader to the PhD thesis of Hopkins~\cite{sam-thesis} or the survey article~\cite{lowdeg-notes}.

Underlying the above ideas is the belief that for certain ``natural'' problems, low-degree statistics are as powerful as all polynomial-time algorithms -- we refer broadly to this belief as the ``low-degree conjecture''.
However, formalizing the notion of ``natural'' problems is not a straightforward task. Perhaps the easiest way to illustrate the meaning of ``natural'' is by example: prototypical examples (studied in the previously mentioned works) include planted clique, sparse PCA, random constraint satisfaction problems, community detection in the stochastic block model, spiked matrix models, tensor PCA, and various problems of a similar flavor. All of these can be stated as simple hypothesis testing problems between a ``null'' distribution (consisting of random noise) and a ``planted'' distribution (which contains a ``signal'' hidden in noise). They are all high-dimensional problems but with sufficient symmetry that they can be specified by a small number of parameters (such as a ``signal-to-noise ratio''). For all of the above problems, the best known polynomial-time algorithms succeed precisely in the parameter regime where simple statistics succeed, i.e., where there exists a $O(\log n)$-degree polynomial of the data whose value behaves noticeably different under the null and planted distributions (in a precise sense). Thus, barring the discovery of a drastically new algorithmic approach, the low-degree conjecture seems to hold for all the above problems. In fact, a more general version of the conjecture seems to hold for runtimes that are not necessarily polynomial: degree-$D$ statistics are as powerful as all $n^{\tilde \Theta(D)}$-time algorithms, where $\tilde \Theta$ hides factors of $\log n$ \cite[Hypothesis~2.1.5]{sam-thesis} (see also~\cite{lowdeg-notes,subexp-sparse}).

A precise version of the low-degree conjecture was formulated in the PhD thesis of Hopkins~\cite{sam-thesis}. This includes precise conditions on the null distribution $\nu$ and planted distribution $\mu$ which capture most of the problems mentioned above. The key conditions are that there should be sufficient symmetry, and that $\mu$ should be injected with at least a small amount of noise. Most of the problems above satisfy this symmetry condition (a notable exception being the spiked Wishart model\footnote{Here we mean the formulation of the spiked Wishart model used in~\cite{sk-cert}, where we directly observe Gaussian samples instead of only their covariance matrix.}, which satisfies a mild generalization of it), but it remained unclear whether this assumption was needed in the conjecture. On the other hand, the noise assumption is certainly necessary, as illustrated by the example of solving a system of linear equations over a finite field: if the equations have an exact solution then it can be obtained via Gaussian elimination even though low-degree statistics suggest that the problem should be hard; however, if a small amount of noise is added (so that only a $1-\varepsilon$ fraction of the equations can be satisfied) then Gaussian elimination is no longer helpful, and the low-degree conjecture seems to hold.

In this work we investigate more precisely what kinds of noise and symmetry conditions are needed in the conjecture of Hopkins~\cite{sam-thesis}. Our first result (Theorem~\ref{thm:refute}) actually refutes the conjecture in the case where the underlying random variables are real-valued. Our counterexample exploits the specifics of the noise operator used in the conjecture, along with the fact that a single real number can be used to encode a large (but polynomially bounded) amount of data. In other words, we show that a \emph{stronger} noise assumption than the one in~\cite{sam-thesis} is needed; Remark~\ref{rem:OU} explains a modification of the conjecture that we do not know how to refute. Our second result (Theorem~\ref{thm:sym}) shows that the symmetry assumption in \cite{sam-thesis} cannot be dropped, i.e., we give a counterexample for a weaker conjecture that does not require symmetry. Both of our counterexamples are based on efficiently decodable error-correcting codes.

\subsection*{Notation}

Asymptotic notation such as $o(1)$ and $\Omega(1)$ pertains to the limit $n \to \infty$. We say that an event occurs with \emph{high probability} if it occurs with probability $1-o(1)$, and we use the abbreviation w.h.p.\ (``with high probability''). We use $[n]$ to denote the set $\{1,2,\ldots,n\}$. The \emph{Hamming distance} between vectors $x,y \in F^n$ (for some field $F$) is $\Delta(x,y) = |\{i \in [n] \,:\, x_i \ne y_i\}|$ and the \emph{Hamming weight} of $x$ is $\Delta(x,0)$.

\section{The Low-Degree Conjecture}

We now state the formal variant of the low-degree conjecture proposed in the PhD thesis of Hopkins~\cite[Conjecture 2.2.4]{sam-thesis}. The terminology used in the statement will be explained below.

\begin{conjecture}\label{conj:sam}
Let $\Omega$ be a finite set or $\RR$, and let $k \ge 1$ be a fixed integer. Let $N = \binom{n}{k}$. Let $\nu$ be a product distribution on $\Omega^N$. Let $\mu$ be another distribution on $\Omega^N$. Suppose that $\mu$ is $S_n$-invariant and $(\log n)^{1+\Omega(1)}$-wise almost independent with respect to $\nu$. Then no polynomial-time computable test distinguishes $T_\delta \mu$ and $\nu$ with probability $1-o(1)$, for any $\delta > 0$. Formally, for all $\delta > 0$ and every polynomial-time computable $t: \Omega^N \to \{0,1\}$ there exists $\delta' > 0$ such that for every large enough $n$,
\[ \frac{1}{2} \PP_{x \sim \nu}(t(x) = 0) + \frac{1}{2} \PP_{x \sim T_\delta \mu}(t(x) = 1) \le 1 - \delta'. \]
\end{conjecture}

\noindent (The asymptotic notation $\Omega(1)$ is not to be confused with the set $\Omega$.) We now explain some of the terminology used in the conjecture, referring the reader to~\cite{sam-thesis} for the full details. We will be concerned with the case $k=1$, in which case $S_n$-invariance of $\mu$ means that for any $x \in \Omega^n$ and any $\pi \in S_n$ (the symmetric group) we have $\PP_\mu(x) = \PP_\mu(\pi \cdot x)$ where $\pi$ acts by permuting coordinates. The notion of \emph{$D$-wise almost independence} captures how well degree-$D$ polynomials can distinguish $\mu$ and $\nu$. For our purposes, we do not need the full definition of $D$-wise almost independence (see~\cite{sam-thesis}), but only the fact that it is implied by \emph{exact} $D$-wise independence, defined as follows.

\begin{definition}
A distribution $\mu$ on $\Omega^N$ is $D$-wise independent with respect to $\nu$ if for any $S \subseteq [N]$ with $|S| \le D$ we have equality of the marginal distributions $\mu|_S = \nu|_S$.
\end{definition}

\noindent Finally, the \emph{noise operator} $T_\delta$ is defined as follows.

\begin{definition}
Let $\nu$ be a product distribution on $\Omega^N$ and let $\mu$ be another distribution on $\Omega^N$. For $\delta \in [0,1]$, let $T_\delta \mu$ be the distribution on $\Omega^N$ generated as follows. To sample $z \sim T_\delta \mu$, first sample $x \sim \mu$ and $y \sim \nu$ independently, and then, independently for each $i$, let
\[ z_i = \left\{\begin{array}{ll} x_i & \text{with probability } 1-\delta, \\ y_i & \text{with probability } \delta. \end{array}\right. \]
(Note that $T_\delta$ depends on $\nu$ but the notation suppresses this dependence; $\nu$ will always be clear from context.)
\end{definition}

\section{Main Results}

We first give a counterexample that refutes Conjecture~\ref{conj:sam} in the case $\Omega = \RR$.

\begin{theorem}\label{thm:refute}
The following holds for infinitely many $n$. Let $\Omega = \RR$, and $\nu = \mathrm{Unif}([0,1]^n)$. There exists a distribution $\mu$ on $\Omega^n$ such that $\mu$ is $S_n$-invariant (with $k=1$) and $\Omega(n)$-wise independent with respect to $\nu$, and for some constant $\delta > 0$ there exists a polynomial-time computable test distinguishing $T_\delta \mu$ and $\nu$ with probability $1-o(1)$.
\end{theorem}

\noindent The proof is given in Section~\ref{sec:pf-refute}.

\begin{remark}
We assume a standard model of finite-precision arithmetic over $\RR$, i.e., the algorithm $t$ can access polynomially-many bits in the binary expansion of its input.
\end{remark}

\noindent Note that we refute an even weaker statement than Conjecture~\ref{conj:sam} because our counterexample has \emph{exact} $\Omega(n)$-wise independence instead of only $(\log n)^{1+\Omega(1)}$-wise \emph{almost} independence.

\begin{remark}\label{rem:OU}
Essentially, our counterexample exploits the fact that a single real number can be used to encode a large block of data, and that the noise operator $T_\delta$ will leave many of these blocks untouched (effectively allowing us to use a super-constant alphabet size). We therefore propose modifying Conjecture~\ref{conj:sam} in the case $\Omega = \RR$ by using a different noise operator that applies a small amount of noise to \emph{every} coordinate instead of resampling a small number of coordinates. If $\nu$ is i.i.d.\ $\mathcal{N}(0,1)$ then the standard Ornstein-Uhlenbeck noise operator is a natural choice (and in fact, this is mentioned by~\cite{sam-thesis}). Formally, this is the noise operator $T_\delta$ that samples $T_\delta \mu$ as follows: draw $x \sim \mu$ and $y \sim \nu$ and output $\sqrt{1-\delta} x + \sqrt{\delta} y$.
\end{remark}

Our second result illustrates that in the case where $\Omega$ is a finite set, the $S_n$-invariance assumption cannot be dropped from Conjecture~\ref{conj:sam}. (In stating the original conjecture, Hopkins~\cite{sam-thesis} remarked that he was not aware of any counterexample when the $S_n$-invariance assumption is dropped).

\begin{theorem}\label{thm:sym}
The following holds for infinitely many $n$. Let $\Omega = \{0,1\}$ and $\nu = \mathrm{Unif}(\{0,1\}^n)$. There exists a distribution $\mu$ on $\Omega^n$ such that $\mu$ is $\Omega(n)$-wise independent with respect to $\nu$, and for some constant $\delta > 0$ there exists a polynomial-time computable test distinguishing $T_\delta \mu$ and $\nu$ with probability $1-o(1)$.
\end{theorem}

\noindent The proof is given in Section~\ref{sec:pf-sym}.

Both of our counterexamples are in fact still valid in the presence of a stronger noise operator $T_\delta$ that \emph{adversarially} changes any $\delta$-fraction of the coordinates.

The rest of the paper is organized as follows. Our counterexamples are based on error-correcting codes, so in Section~\ref{sec:coding} we review the basic notions from coding theory that we will need. In Section~\ref{sec:proofs} we construct our counterexamples and prove our main results.

\section{Coding Theory Preliminaries}
\label{sec:coding}

Let $F = \FF_q$ be a finite field. A \emph{linear code} $C$ (over $F$) is a linear subspace of $F^n$. Here $n$ is called the \emph{(block) length}, and the elements of $C$ are called \emph{codewords}. The \emph{distance} of $C$ is the minimum Hamming distance between two codewords, or equivalently, the minimum Hamming weight of a nonzero codeword.

\begin{definition}
Let $C$ be a linear code. The \emph{dual distance} of $C$ is the minimum Hamming weight of a vector in $F^n$ that is orthogonal to all codewords. Equivalently, the dual distance is the distance of the \emph{dual code} $C^\perp = \{x \in F^n \,:\, \langle x,c \rangle = 0 \;\;\forall c \in C\}$.
\end{definition}

\noindent The following standard fact will be essential to our arguments.

\begin{proposition}\label{prop:dual-k-wise}
If $C$ is a linear code with dual distance $d$, then the uniform distribution over codewords is $(d-1)$-wise independent with respect to the uniform distribution on $F^n$.
\end{proposition}

\begin{proof}
This is a standard fact in coding theory, but we give a proof here for completeness. Fix $S \subseteq [n]$ with $|S| \le d-1$. For some $k$, we can write $C = \{x^\top G \,:\, x \in F^k\}$ for some $k \times n$ \emph{generator matrix} $G$ whose rows form a basis for $C$. Let $G_S$ be the $k \times |S|$ matrix obtained from $G$ by keeping only the columns in $S$. It is sufficient to show that if $x$ is drawn uniformly from $F^k$ then $x^\top G_S$ is uniform over $F^{|S|}$. The columns of $G_S$ must be linearly independent, because otherwise there is a vector $y \in F^n$ of Hamming weight $\le d-1$ such that $Gy = 0$, implying $y \in C^\perp$, which contradicts the dual distance. Thus there exists a set $T \subseteq [k]$ of $|S|$ linearly independent rows of $G_S$ (which form a basis for $F^{|S|}$). For any fixed choice of $x_{[k] \setminus T}$ (i.e., the coordinates of $x$ outside $T$), if the coordinates $x_T$ are chosen uniformly at random then $x^\top G_S$ is uniform over $F^{|S|}$. This completes the proof.
\end{proof}

\begin{definition}
  A code $C$ admits efficient decoding from $r$ errors and $s$ erasures if there is a deterministic polynomial-time algorithm $\mathcal{D}: (F \cup \{\perp\})^n \to F^n \cup \{\mathsf{fail}\}$ with the following properties.
  \begin{itemize}
  \item For any codeword $c \in C$, let $c' \in F^n$ be any vector obtained from $c$ by changing the values of at most $r$ coordinates (to arbitrary elements of $F$), and replacing at most $s$ coordinates with the erasure symbol $\perp$. Then $\mathcal{D}(c') = c$.
  \item For any arbitrary $c' \in F^n$ (not obtained from some codeword as
    above), $\mathcal{D}(c')$ can output any codeword or $\mathsf{fail}$ (but
    must never output a vector that is not a codeword\footnote{The codes we
      deal with in this paper can be efficiently constructed, and so it is
      easy to test whether a given vector is a codeword. Thus, this assumption
      is without loss of generality.}).
  \end{itemize}
\end{definition}
\noindent Note that the decoding algorithm knows where the erasures have occurred but does not know where the errors have occurred.

Our first counterexample (Theorem~\ref{thm:refute}) is based on the classical \emph{Reed-Solomon codes} $\mathsf{RS}_q(n, k)$, which consist of univariate polynomials of degree at most $k$ evaluated at $n$ canonical elements of the field $\FF_q$, and are known to have the following properties.

\begin{proposition}[Reed-Solomon Codes]
\label{prop:RS}
For any integers $0 \le k < n$ and for any prime power $q \ge n$, there is a length-$n$ linear code $C$ over $\FF_q$ with the following properties:
\begin{itemize}
    \item the dual distance is $k+2$,
    \item $C$ admits efficient decoding from $r$ errors and $s$ erasures whenever $2r+s < n-k$.
\end{itemize}
\end{proposition}

\begin{proof}
See e.g., \cite{GS-RS} for the construction and basic facts regarding Reed-Solomon codes $\mathsf{RS}_q(n,k)$. The distance of $\mathsf{RS}_q(n,k)$ is $n-k$. It is well known that the dual code of $\mathsf{RS}_q(n,k)$ is $\mathsf{RS}_q(n,n-k-2)$, which proves the claim about dual distance. Efficient decoding is discussed e.g., in Section~3.2 of~\cite{GS-RS}.
\end{proof}

\noindent Our second counterexample (Theorem~\ref{thm:sym}) is based on the following construction of efficiently correctable binary codes with large dual distance. A proof (by Guruswami) can be found in~\cite[Theorem~4]{good-dual-codes}. Similar results were proved earlier~\cite{SV-code,FR-code,GS-code}.

\begin{proposition}
\label{prop:good-dual}
There exists a universal constant $\zeta \ge 1/30$ such that for every integer
$i \ge 1$ there is a linear code $C$ over $\FF_2 = \{0,1\}$ of block length
$n = 42 \cdot 8^{i+1}$, with the following properties:
\begin{itemize}
    \item the dual distance is at least $\zeta n$, and
    \item $C$ admits efficient decoding from $\zeta n/2$ errors (with no erasures).
\end{itemize}
\end{proposition}

\section{Proofs}
\label{sec:proofs}

Before proving the main results, we state some prerequisite notation and lemmas.

\begin{definition}
Let $F$ be a finite field. For $S \subseteq[n]$ (representing erased positions), the \emph{$S$-restricted Hamming distance} $\Delta_S(x,y)$ is the number of coordinates in $[n] \setminus S$ where $x$ and $y$ differ:
\[ \Delta_S(x,y) = |\{i \in [n] \setminus S \,:\, x_i \ne y_i\}|. \]
We allow $x,y$ to belong to either $F^n$ or $F^{[n] \setminus S}$, or even to $(F \cup \{\perp\})^n$ so long as the ``erasures'' $\perp$ occur only in $S$.
\end{definition}

\noindent The following lemma shows that a random string is sufficiently far from any codeword.

\begin{lemma}\label{lem:rand-far}
Let $C$ be a length-$n$ linear code over a finite field $F$. Suppose $C$ admits efficient decoding from $2r$ errors and $s$ erasures, for some $r,s$ satisfying $r \le (n-s)/(8e)$. For any fixed choice of at most $s$ erasure positions $S \subseteq [n]$, if $x$ is a uniformly random element of $F^{[n] \setminus S}$ then
\[ \mathbb{P}_x(\exists c \in C \,:\, \Delta_S(c,x) \le r) \le (r+1) 2^{-r}. \]
\end{lemma}

\begin{proof}
Let $B_r(c) = \{x \in F^{[n] \setminus S} \,:\, \Delta_S(c,x) \le r\} \subseteq F^{[n] \setminus S}$ denote the Hamming ball (with erasures $S$) of radius $r$ and center $c$, and let $|B_r|$ denote its cardinality (which does not depend on $c$). We have the following basic bounds on $|B_r|$:
\[ \binom{n-|S|}{r}(|F|-1)^r \le |B_r| \le (r+1)\binom{n-|S|}{r}(|F|-1)^r. \]
Since decoding from $2r$ errors and $s$ erasures is possible, the Hamming balls $\{B_{2r}(c)\}_{c \in C}$ are disjoint, and so
\begin{align*}
\mathbb{P}_x(\exists c \in C \,:\, \Delta_S(c,x) \le r) &\le |B_r|/|B_{2r}| \\
&\le (r+1) \binom{n-|S|}{r} \binom{n-|S|}{2r}^{-1} (|F|-1)^{-r} \\
&\le (r+1) \binom{n-|S|}{r} \binom{n-|S|}{2r}^{-1}.
\intertext{Using the standard bounds $\left(\frac{n}{k}\right)^k \le \binom{n}{k} \le \left(\frac{ne}{k}\right)^k$, this becomes}
&\le (r+1) \left(\frac{4e r}{n-|S|}\right)^r \\
&\le (r+1) \left(\frac{4e r}{n-s}\right)^r
\end{align*}
which is at most $(r+1)2^{-r}$ provided $r \le (n-s)/(8e)$.
\end{proof}

\begin{lemma}\label{lem:unique}
Let $j_1,\ldots,j_n$ be uniformly and independently chosen from $[n]$. For any constant $\alpha < 1/e$, the number of indices $i \in [n]$ that occur exactly once among $j_1,\ldots,j_n$ is at least $\alpha n$ with high probability.
\end{lemma}

\begin{proof}
Let $X_i \in \{0,1\}$ be the indicator that $i$ occurs exactly once among $j_1,\ldots,j_n$, and let $X = \sum_{i=1}^n X_i$. We have (as $n \to \infty$)
\[ \EE[X_i] = n(1/n)(1-1/n)^{n-1} \to e^{-1}, \]
$\EE[X_i^2] = \EE[X_i]$, and for $i \ne i'$,
\[ \EE[X_i X_{i'}] = n(n-1)(1/n)^2(1-2/n)^{n-2} \to e^{-2}. \]
This means $\EE[X] = (1+o(1)) n/e$ and
\begin{align*}
\mathrm{Var}(X) &= \sum_i (\EE[X_i^2] - \EE[X_i]^2) + \sum_{i \ne i'} (\EE[X_i X_{i'}] - \EE[X_i] \EE[X_{i'}]) \\
&\le n(1+o(1))(e^{-1} - e^{-2}) + n(n-1) \cdot o(1) \\
&= o(n^2).
\end{align*}
The result now follows by Chebyshev's inequality.
\end{proof}

\subsection{Proof of Theorem~\ref{thm:refute}}
\label{sec:pf-refute}

The idea of the proof is as follows. First imagine the setting where $\mu$ is
not required to be $S_n$-invariant. By using each real number to encode an
element of $F = \FF_q$, we can take $\nu$ to be the uniform distribution on
$F^n$ and take $\mu$ to be a random Reed-Solomon codeword in $F^n$. Under
$T_\delta \mu$, the noise operator will corrupt a few symbols (``errors''), but
the decoding algorithm can correct these and thus distinguish $T_\delta \mu$
from $\nu$.

In order to have $S_n$-invariance, we need to modify the construction. Instead of observing an ordered list $y = (y_1,\ldots,y_n)$ of symbols, we will observe pairs of the form $(i,y_i)$ (with each pair encoded by a single real number) where $i$ is a random index. If the same $i$ value appears in two different pairs, this gives conflicting information; we deal with this by simply throwing it out and treating $y_i$ as an ``erasure'' that the code needs to correct. If some $i$ value does not appear in any pairs, we also treat this as an erasure. The full details are given below.

\begin{proof}[Proof of Theorem~\ref{thm:refute}]
Let $C$ be the length-$n$ code from Proposition~\ref{prop:RS} with $k = \lceil \alpha n \rceil$ for some constant $\alpha \in (0,1)$ to be chosen later. Fix a scheme by which a real number encodes a tuple $(j,y)$ with $j \in [n]$ and $y \in F = \FF_q$, in such a way that a uniformly random real number in $[0,1]$ encodes a uniformly random tuple $(j,y)$. More concretely, we can take $n = q = 2^m$ for some integer $m \ge 1$, in which case $(j,y)$ can be directly encoded using the first $2m$ bits of the binary expansion of a real number. Under $x \sim \nu$, each coordinate $x_i$ encodes an independent uniformly random tuple $(j_i,y_i)$. Under $\mu$, let each coordinate $x_i$ be a uniformly random encoding of $(j_i,y_i)$, drawn as follows. Let $\tilde c$ be a uniformly random codeword from $C$. Draw $j_1,\ldots,j_n \in [n]$ independently and uniformly. For each $i$, if $j_i$ is a unique index (in the sense that $j_i \ne j_{i'}$ for all $i' \ne i$) then set $y_i = \tilde c_{j_i}$; otherwise choose $y_i$ uniformly from $F$.

Note that $\mu$ is $S_n$-invariant. Since the dual distance of $C$ is
$k+2 = \Omega(n)$, it follows (using Proposition~\ref{prop:dual-k-wise}) that
$\mu$ is $\Omega(n)$-wise independent with respect to $\nu$. By choosing
$\delta > 0$ and $\alpha > 0$ sufficiently small, we can ensure that
$16 \delta n + (2n/3 + 4\delta n) < n - k$ and so $C$ admits efficient decoding
from $8 \delta n$ errors and $2n/3 + 4\delta n$ erasures (see
Proposition~\ref{prop:RS}). The algorithm to distinguish $T_\delta \mu$ and
$\nu$ is as follows. Given a list of $(j_i,y_i)$ pairs, produce $c' \in F^n$ by
setting $c'_{j_i} = y_i$ wherever $j_i$ is a unique index (in the above sense),
and setting all other positions of $c'$ to $\bot$ (an ``erasure''). Let
$S \subseteq [n]$ be the indices $i$ for which $c'_i = \perp$. Run the decoding
algorithm on $c'$; if it succeeds and outputs a codeword $c$ such that
$\Delta_S(c,c') \le 4 \delta n$ then output ``$T_\delta \mu$'', and otherwise
output ``$\nu$''.

We can prove correctness as follows. If the true distribution is $\nu$ then Lemma~\ref{lem:unique} guarantees $|S| \le 2n/3$ (w.h.p.). Since the non-erased values in $c'$ are uniformly random, Lemma~\ref{lem:rand-far} ensures there is no codeword $c$ with $\Delta_S(c,c') \le 4\delta n$ (w.h.p.), provided we choose $\delta \le 1/(96e)$, and so our algorithm outputs ``$\nu$'' (w.h.p).

Now suppose the true distribution is $T_\delta \mu$. In addition to the
$\le 2n/3$ erasures caused by non-unique $j_i$'s sampled from $\mu$, each
coordinate resampled by $T_\delta$ can create up to $2$ additional erasures and
can also create up to $2$ errors (i.e., coordinates $i$ for which $c_i' \ne \perp$ but
$c_i' \ne \tilde c_i$). Since at most $2\delta n$ coordinates get resampled
(w.h.p.), this means we have a total of (up to) $4\delta n$ errors and
$2n/3 + 4\delta n$ erasures.  This means decoding succeeds and outputs
$\tilde c$ (i.e., the true codeword used to sample $\mu$), and furthermore,
$\Delta_S(\tilde c,c') \le 4 \delta n$.
\end{proof}

\subsection{Proof of Theorem~\ref{thm:sym}}
\label{sec:pf-sym}

The idea of the proof is similar to the previous proof, and somewhat simpler (since we do not need $S_n$-invariance). We take $\nu$ to be the uniform distribution on binary strings and take $\mu$ to be the uniform distribution on codewords, using the binary code from Proposition~\ref{prop:good-dual}. The decoding algorithm is able to correct the errors caused by $T_\delta$.

\begin{proof}[Proof of Theorem~\ref{thm:sym}]
Let $C$ be the code from Proposition~\ref{prop:good-dual}. Each codeword $c \in C$ is an element of $\FF_2^n$, which can be identified with $\{0,1\}^n = \Omega^n$. Let $\mu$ be the uniform distribution over codewords. Since the dual distance of $C$ is at least $\zeta n$, we have from Proposition~\ref{prop:dual-k-wise} that $\mu$ is $\Omega(n)$-wise independent with respect to the uniform distribution $\nu$. We also know that $C$ admits efficient decoding from $\zeta n/2$ errors.

Let $\delta = \min\{1/(16e),\zeta/8\}$. The algorithm to distinguish $T_\delta \mu$ and $\nu$, given a sample $c'$, is to run the decoding algorithm on $c'$; if decoding succeeds and outputs a codeword $c$ such that $\Delta(c,c') \le 2 \delta n$ then output ``$T_\delta \mu$'', and otherwise output ``$\nu$''.

We can prove correctness as follows. If $c'$ is drawn from $T_\delta \mu$ then $c'$ is separated from some codeword $c$ by at most $2\delta n \le \zeta n/4$ errors (w.h.p.), and so decoding will find $c$. If instead $c'$ is drawn from $\nu$ then (since $\delta \le 1/(16e)$) by Lemma~\ref{lem:rand-far}, there is no codeword within Hamming distance $2\delta n$ of $c'$ (w.h.p.).
\end{proof}

\section*{Acknowledgments}

We thank Sam Hopkins and Tim Kunisky for comments on an earlier draft.

\bibliographystyle{alpha}
\bibliography{main}

\end{document}